\documentclass[11pt]{article}
\usepackage{latexsym, amsmath, amsthm, amsfonts, amssymb, enumitem,
  color, xspace, graphicx, ifthen}
\usepackage[sort&compress,square,numbers]{natbib}
\usepackage[colorlinks]{hyperref}
\usepackage[utf8]{inputenc}

\usepackage{environ}
\makeatletter
\NewEnviron{tagleft}{\tagsleft@true\BODY}  
\NewEnviron{tagright}{\tagsleft@false\BODY}
\NewEnviron{hide}{}                        
\makeatother

\graphicspath{{images-crop/}}

\definecolor{BLUE}{rgb}{0,0,1}
\definecolor{RED}{rgb}{1,0,0}
\definecolor{GREEN}{rgb}{0,1,0}
\definecolor{PURPLE}{rgb}{1,0,1}

\def\RR{\mathbb{R}}

\newtheorem{thm}{Theorem}[section]
\newtheorem{cor}[thm]{Corollary}
\newtheorem{lem}[thm]{Lemma}
\newtheorem{exmp}[thm]{Example}

\def\bm#1{\boldsymbol{#1}}

\def\Cbr#1{\left\{#1\right\}}   

\def\FDP{\mathrm{FDP}}
\def\NTD{\mathrm{NTD}}

\def\SE{\mathrm{SE}}

\def\cov{\mathrm{Cov}}          
\def\mean{\mathrm{E}}           
\def\pr{\mathrm{P}}             
\def\lfrac#1#2{#1/#2}
\makeatletter
\def\iid{i.i.d\@ifnextchar.{}{.}}
\makeatother

\def\ex{\mathrm{e}}     
\def\diag{\mathrm{diag}}        
\def\toi{\to\infty}
\def\tr{\mathrm{Tr}}            
\def\tp{\sp{\mathrm t}}              

\def\trf{\varphi}               
\def\gv{\,|\,}                  
\def\rx{\epsilon}

\def\eim{\varrho}                 
\def\slem{\eim_2}           
\def\tlem{\eim_3}           
\def\psig{\psi}                  

\def\enum#1#2{#1_1, \ldots, #1_#2}
\def\cum#1#2{#1_1 + \ldots + #1_#2}

\def\showfig#1{\scalebox{.5}{\includegraphics{#1.pdf}}}

\begin{document}
\begin{center}
  \textbf{
    \uppercase{
     False discovery variance reduction in large scale simultaneous hypothesis tests
    }
  }
  \\[2ex]
  Sairam Rayaprolu and Zhiyi Chi\footnote{Department of Statistics,
    University of Connecticut, 215 Glenbrook Road, U-4120, Storrs, CT
    06269, sairamray@gmail.com and zhiyi.chi@uconn.edu.}
\end{center}

\begin{abstract}
  Statistical dependence between hypotheses poses a significant
  challenge to the stability of large scale multiple hypotheses
  testing.  Ignoring it often results in an unacceptably large
    spread in the false positive proportion even though the average
    value is acceptable \cite{owen_2005,
    fanhangu_2012, qiukleyak_2005,schlin_2011}.  However, the statistical
  dependence structure of data is often unknown.  Using a generic
  signal-processing model, Bayesian multiple testing, and
  simulations, we demonstrate that the variance of the false positive
  proportion can be substantially reduced even under unknown short
  range dependence.  We do this by modeling the data generating
  process as a stationary ergodic binary signal process embedded in
  noisy observations.  We derive conditional probabilities needed for
  the Bayesian multiple testing by incorporating nearby
  observations into a second order Taylor series approximation.
  Simulations under general conditions are carried out to assess
  the validity and the variance reduction of the approach.
  Along the  way, we address the problem of sampling a
  random Markov matrix with specified stationary distribution and
  lower bounds on the top absolute eigenvalues, which is of interest
  in its own right.

  \medbreak\noindent

  \medbreak\noindent
\end{abstract}

\section{Introduction}\label{sec:intro}
The problem of multiplicity in simultaneous hypothesis testing has
been well recognized in experimental research and in statistics for
more than half a century.  However, contemporary technologies have
multiplied the scale of the problem by several orders of magnitude.
Whereas classical problems were concerned with simultaneous inference
of several or tens of hypotheses, present-day data intensive
technologies such as DNA microarray in genomics and functional
magnetic resonance imaging (fMRI) in brain imaging routinely produce
tens or even hundreds of thousands of simultaneous hypotheses.

To get an idea on the severity of the problem of multiplicity,
consider, for example, fMRI used to detect brain areas that activate
in response to a specific task.  A single fMRI can contain hundreds of
thousands of multivariate observations, each being a time series of
measurements at a spatial location, or ``voxel'',  in the brain
\cite{linmej_2015}.  A typical analysis on such a dataset uses a
``massively univariate approach''.  For each voxel, a separate
hypothesis test is performed on the time series of its measurements
with the null  hypothesis being that there is no activation in the
voxel.   Even with no brain activation in any voxel, if there are, say
100,000 voxels,   a voxel-wise significance level of $0.05$ would
result in roughly 5000   significant $p$-values, in other words, 5000
false positives.  Therefore it is necessary to choose a more stringent,
sufficiently small $p$-value threshold.  However, a $p$-value
threshold that is too stringent will result in unacceptably many false
negatives often resulting in failure to detect any voxels with
significant brain activation.  Determining a $p$-value threshold that
strikes a balance between the two extremes is at the heart of the
large scale multiple testing problem.  See \cite{farcomeni_2008, roquain_2011}
for a broad overview of research in multiple testing.

False positives are usually referred to as \emph{false discoveries} in
multiple testing literature.  The outcomes of multiple testing are
summarized in the table below.  Here, $V$ is the number of Type I
errors (false discoveries), $T$ the number of false non-discoveries
(false negatives), $S$ the number of true discoveries (NTD) or true
positives and $U$ the number of true non-discoveries (true negatives).
$R$ is the total number of discoveries (positives) and $A$ the total
number of non-discoveries (negatives).  Finally $m_0$ is the  number
of true nulls and $m$ the total number of hypotheses.
\begin{center}
  \begin{tabular}{|l|ccr|}
    \hline
    Hypothesis \rule{0em}{1.2em}& Accept & Reject & Total
    \\\hline
    True null \rule{0em}{1.2em}& $U$ & $V$ & $m_0$\\
    False null& $T$ & $S$ & $m_1$ \\
    & $A$ & $R$ & $m$ \\\hline
  \end{tabular}
\end{center}

Procedures that reliably control Type I errors have been a focus in
large scale multiple testing.  To start with, a numerical criterion on
the testing result has to be specified.  Since its introduction by
Benjamini and Hochberg \cite{benhoc_1995}, the false discovery rate
(FDR) has been widely used.  Consequently, a great deal of research
has been focused in particular on procedures that tend to have a low FDR
irrespective of the distribution of data or the application.

By definition, the FDR is the expectation of $\lfrac V{\max\!\Cbr{R,
    1}}$.  The latter, known as the false discovery  proportion (FDP),
is a random variable and summarizes the Type I errors of the multiple
testing on the data at hand.  When its variance is high there will be
a high  probability of it being substantially larger than the FDR,
making the testing result less reliable even if \emph{on average\/}
it is at an acceptable level.  Since it is usually not feasible to
improve the reliability by repeating a large scale
multiple testing procedure numerous times, it is important to achieve
low variability of the FDP.  In a similar spirit, a performance
indicator of practical importance is \emph{stability}, defined as the
the standard deviation of the total number of discoveries
\cite{gorglaqiuyak_2007, li_2015}.

Assuming independence of test statistics of the hypotheses, the FDR
can be controlled by the well-known Benjamini-Hochberg (BH) procedur
e
\cite {benhoc_1995, stotaysie_2004}.  However, in practice the data
in large scale multiple testing almost always has certain
statistical dependence.  Under various assumptions of dependence, the
BH procedure is still valid, i.e., able the bound the FDR at a desired
level \cite{wu_2008, benyek_2001, sarkar_2002, sarkar_2006,
  clahal_2009}.  On the other hand, the FDR is mostly, though not
entirely, concerned with Type I errors.  There are other important
metrics to evaluate a multiple testing procedure.  Besides the
aforementioned metrics on variability, a key metric is \emph{power},
defined as $\mean(S)$ or $\mean(S/m_1)$.  When several aspects are
taken into account, in the presence of substantial dependence between
the nulls, the overall performance of the BH procedure deteriorates,
sometimes severely \cite {efron_2007,  fanhangu_2012, owen_2005, qiukleyak_2005, roqvil_2011}.

To tackle statistical dependence in multiple testing, some
studies have taken a probabilistic approach by modeling the data as
stochastic processes with nontrivial dependence structure.  They use
specific models such as hidden Markov models \cite{suncai_2009,
  wu_2008}, Markov random fields \cite{ngumclche_2014,liuzhapag_2016}, or Gaussian
random fields \cite{chesch_2017, schgavadl_2011}. The main strength of
these models is their tractability.  However, the
statistical dependence structure of a dataset is often unknown or only
partially known, which is in fact a major obstacle to multiple
testing.  For example, recent research has shown that that the
principal cause of the invalid cluster inferences in fMRI is spatial
autocorrelation functions that do not follow the assumed Gaussian
shape \cite{eknikn_2017}.

In this paper, we consider the FDR control with acceptable variability
in the  FDP, when the knowledge about the dependence structure of the
nulls is very limited.  We propose a Bayesian approach that is based
on a generic signal-processing model and demonstrate through
simulations that the variance of the FDP can be substantially reduced
even under unknown short range dependence.  While the nulls
are modeled by a hidden process, no assumptions are made on the
dependence structure of the process.  Thus, the approach is
nonparametric regarding the dependence between the nulls.  Under the
Bayesian setting, it is known that multiple testing based on the
conditional probabilities of the nulls is optimal under certain
criteria (\cite{chi_2011, sarzhogho_2008}; cf.\ Section
\ref{s:relatedwork}).  However, the conditional probabilities are not
available due to the unknown dependence structure of the nulls.  The
idea of our approach is to treat the conditional probabilities as
functions of the signal strength and approximates them by a second
order Taylor series in the weak-signal regime.  In addition to
providing much technical convenience, consideration of the weak-signal
regime is of practical relevance due to the challenge it poses to
multiple testing for very noisy data \cite {linmej_2015}.  The
resulting approximation has two novel features.  First, it naturally
incorporates clusters of neighboring observations.  Second, it does
this without knowledge of the joint distribution of neighboring nulls.
Instead, it utilizes their empirical moments.  Expansion in terms of
moments or cumulants is a defining feature in well known
approximations such as normal approximation and Edgeworth
approximation.  It is interesting that our approximation has a similar
feature, even though it results from a Taylor expansion in a parameter
not directly related to the hidden process of the nulls.

After a brief literature review in Section \ref{s:relatedwork}, in
Section \ref{s:usage}, we introduce the signal-processing model used
in the Bayesian approach and then derive the second order Taylor
approximations to the conditional probabilities.  In Section
\ref{s:numerical}, we report numerical studies to assess multiple
testing based on approximated conditional likelihoods in the presence
of dependence.  The results show that the multiple testing can
  control the FDR while having a much smaller variance in the FDP
  comparing to the BH procedure.  To conduct the numerical study, we
have to address the problem of sampling a random Markov matrix with
specified stationary distribution and   lower bounds on the top
absolute eigenvalues.  This sampling issue is of interest in its own
right.  Finally, Section \ref{s:conclusions} makes some concluding
remarks.

\section{Related works and basic setup} \label{s:relatedwork}
There is now a large literature to address statistical dependence in
multiple testing \cite{fanhan_2017, fanhangu_2012, schlin_2011,
  efron_2010, suncai_2009,clahal_2009, leesto_2008, efron_2007,
  sarkar_2006, owen_2005, qiukleyak_2005, sarkar_2002, benyek_2001,
stotaysie_2004, wu_2008, genwas_2004}.
Sun and Cai \cite {suncai_2009}
proposed to exploit the dependence structure of hypotheses from a
decision-theoretical point of view and considered a Hidden Markov
Model (HMM).  Let $\enum H m$ be a set of hypotheses and $\eta_j=0$ if
$H_j$ is true, and 1 otherwise.  Assume that $\bm\eta = (\enum \eta
m)$ form a Markov chain and $\bm X = (\enum X m)$ is a set of
observations that are independent conditional on $\bm\eta$.  Then they
showed that a thresholding procedure for the conditional likelihoods
$\pr(H_j \text{ is false}\gv\bm X) = \pr(\eta_j=1\gv\bm X)$ minimizes
$\mean(\lambda V + T)$ for some $\lambda>0$.  From a Bayesian
perspective, the Markov process is a prior on $\bm\eta$ and
$\pr(\eta_j=1\gv\bm X)$ the posterior likelihood for $H_j$ being true.
In an earlier work \cite{sarzhogho_2008}, it was shown that for any
prior on $\bm\eta$ and any type of data $\bm X$, the following BH-type
procedure, which will be referred to as the Bayes BH procedure,
controls the FDR at a fixed target level $\alpha\in (0,1)$.

\begin{center}
  \textbf{Bayes BH Procedure}\\[2ex]
  \begin{minipage}{.8\textwidth}
    \begin{enumerate}[itemsep=0ex,parsep=0ex]
    \item Sort $\pr(\eta_j=1\gv\bm X)$ in ascending order as $q_1 \le
      q_2 \le \ldots \le q_m$.
    \item If $q_m < 1- \alpha$, then accept all $H_j$, otherwise
      reject all $H_j$ with
      \begin{align*}
        \pr(\eta_j=1\gv \bm X)\ge q_R,
      \end{align*}
      where  $R = \max\{1 \le k \le m: k^{-1}\sum^k_{j=1} q_{m - j + 1}\ge
      1-\alpha \}$.
    \end{enumerate}
  \end{minipage}
\end{center}

It was shown in \cite{chi_2011} that the Bayes BH procedure is power
optimal \emph{a posteriori\/}, that is, it has the largest
$\mean(S\gv\bm X)$ among all procedures with $\mean(\FDP \gv \bm X)\le
\alpha$.  For comparison and later use, the BH procedure is based on
$p$-values of hypotheses and can be described as follows
\cite{simes_1986, benhoc_1995, stotaysie_2004}.

\begin{center}
  \textbf{BH Procedure}\\[2ex]
  \begin{minipage}{.8\textwidth}
    \begin{enumerate}[itemsep=0ex, parsep=0ex]
    \item Sort the $p$-values in ascending order as $p_1 \le p_2 \le
      \ldots \le p_m$.

    \item If $p_1>\alpha/m$, then accept all $H_j$, otherwise reject
      all $H_j$ with $p_j\le p_R$, where $R = \max\{1 \le k \le m : p_k
      \le k \alpha/m\}$.
    \end{enumerate}
  \end{minipage}
\end{center}

The above results highlight the importance of $\pr(\eta_j=1\gv \bm
X)$ in multiple testing.  Once they are computed, the Bayes BH
procedure can be immediately applied.  However, to precisely evaluate
the conditional likelihoods either by closed-form calculation or
simulation, it is necessary to elaborate on the joint dependence
structure of $\bm\eta$ and $\bm X$, which is typically impossible.  In
addition, from a signal processing point of view, multiple testing
becomes useful typically when signals are relatively weak.   We
therefore wish to approximate $\pr(\eta_j=1\gv\bm X)$ using partial
information on the dependence structure that is easy to infer, with
reasonable accuracy especially when the signal is relatively weak.

In \cite{chi_2011}, under the HMM model, $\pr(\eta_j\gv\bm X)$ was
approximated by a second order Taylor series in signal strength.
The Taylor series only uses the the joint distribution of
  triples of nulls.  It turns out that under suitable conditions,
  these distributions can be nonparametrically estimated provided that
  a sample of uncontaminated signals is available.  The Markovian
  structure is not necessary for the Taylor series approximation.
This observation is the starting point of the development in the
following sections.

\section{Approximate conditional likelihoods for multiple
  testing} \label{s:usage}
As mentioned at the end of last section, our Bayesian approach is to
apply the Bayes BH procedure to the conditional likelihoods of
false nulls given the data.  In this section, we first describe a
generic signal processing framework for the Bayesian multiple testing.
Then we derive a second order Taylor series to approximate the
conditional  likelihoods.  Finally, we
specialize to the case where the data consists of observations that
can be described as outcomes of localized interactions between
attenuated signals and noise.  We assume that the parametric form of
the signal-noise interaction is known.  However, throughout the
derivation, no parametric form of the joint distribution of the
signals is assumed.

\subsection{Signal processing point of view of multiple testing}
Let $(H_t)_{t\in T}$ be a set of hypotheses indexed by a finite set
$T$, which can be as simple as $\{1, \ldots, m\}$, or endowed with
certain topology to incorporate, say spatial, relationships between
the  hypotheses.  The ``signal'' is a set of Bernoulli random
variables indexed by $T$, henceforth denoted by $\bm\eta = (\eta_t,\in
T)$.   For each $t\in T$, let $\eta_t = 0$ if $H_t$ is true and 1
otherwise.  The values of $\eta_t$, however, are unobservable, and the
task is to determine which $H_t$ are false, or equivalently, which
$\eta_t$ are equal to 1.  In contrast to \cite{suncai_2009,chi_2011},
no dependence structure of $\bm\eta$ is assumed.

On the other hand, let $\bm X$ denote a vector of observations.  In
the most general setting, $\bm X$ can be of any type.  We assume that
conditional on $\bm\eta$, $\bm X$ has a probability density of the form
\begin{align}  \label{signal-density}
  \rho(\bm x\gv \bm\eta) = \ex^{q(\bm x, \rx\bm\eta)},
\end{align}
where the function $q$ is assumed to be known.  The parameter $\rx>0$
in \eqref{signal-density} will be referred to as ``signal strength''.
Observe that as $\rx=0$, $\bm X$ is independent of $\bm\eta$, and
hence contains no information about the hypotheses.  Under the
setting, for each $t\in T$, the conditional likelihood of $H_t$ being
true given the data $\bm X$ is $\pr(\eta_t=1\gv \bm X)$.  Thus, making
decisions on the nulls is equivalent to recovering $\bm\eta$ from the
data $\bm X$.

\medbreak\noindent
{\it Notation.}  We will treat each $(x_t)_{t\in T}$ as a
column vector.  Thus, for $\bm u = (u_t)_{t\in
  T}$ and $\bm v =   (v_t)_{t\in T}$, $\bm u\tp \bm v= \sum_{t\in T}
u_t v_t$.  If $p$ is a function on $\RR$, then $p(\bm x)$ is a
shorthand for $(p(x_t))_{t\in T}$.  $(M_{st})_{s,t\in
  T}$ denotes a matrix with rows indexed by $s$ and columns indexed by
$t$.  Finally, $\diag(\bm x)$ denotes the diagonal matrix
$(M_{st})_{s,t\in T}$ with $M_{t t}=x_t$ and $M_{s t}=0$ for $s\ne t$.

\subsection{Taylor series approximation of conditional likelihoods}
We next derive an approximation to $\pr(\eta_t=1\gv\bm X)$ using a
second order Taylor series in $\rx$.  The method, however, can
attain higher order Taylor series as well.  The approximation is
valid when the signal strength $\rx$ is low, precisely when the power
of multiple testing suffers the most.  While our objective is to
approximate $\pr(\eta_t = 1\gv\bm X)$, the conditional log-likelihood
ratios
\begin{hide}When the effect of the
hypotheses on the data is ``localized'', i.e., $\bm X = (X_t)_{t\in
  T}$ with
$X_t$ being independent conditionally on $\bm\eta$, the approximation
takes a simpler form.  It only needs as inputs the first three moments
of $\bm\eta$ and the distribution of $X_t$ conditional on $\eta_t$ for
individual $t$, as opposed to information on how the whole $\bm X$
depends on $\bm\eta$.  We will argue that in many settings of signal
processing, both the moments of $\bm\eta$ and the local dependence may
be inferred with reasonable precision.  Results of simulations and
computations for the approximated conditional likelihoods are
presented in Section \ref{s:numerical}.
\end{hide}
\begin{align} \label{e:cllr}
  r_t(\bm X)
  =
  \ln \frac{\pr(\eta_t=1\gv\bm X)}{\pr(\eta_t=0\gv\bm X)}, \quad t\in
  T
\end{align}
are more convenient to work with.  Each $r_t(\bm X)$ is simply the
logit transform of $\pr(\eta_t=1\gv\bm X)$.  Conversely, the latter is
the logistic transform of the former, i.e.,
\begin{align} \label{e:logistic}
  \pr(\eta_t=1\gv\bm X)
  =
  \frac{1}{1+ \exp\{-r_t(\bm X)\}}.
\end{align}

Henceforth, quantities derived from taking expectation conditional on
$\eta_t = i$ will be indexed by $i$ and $t$, e.g.,
\begin{align*}
  \mean_{i t}(\cdot) = \mean(\cdot\gv \eta_t=i),
  \quad
  \rho_{i t}(\cdot) = \rho(\cdot\gv \eta_t = i),
  \quad
  \cov_{i t}(\cdot) = \cov(\cdot\gv \eta_t = i).
\end{align*}
By Bayes rule, for $t\in T$ and $i=0,1$, $\pr(\eta_t=i\gv\bm X)
\propto \pr(\eta_t = i)\rho_{it}(\bm X)$.  Then
\begin{align}  \label{lboldsymbolayes}
  r_t(\bm X)
  =
  \ln\frac{\pr(\eta_t = 1)}{\pr(\eta_t = 0)}+
  \ln\frac{\rho_{1 t}(\bm X)}{\rho_{0 t}(\bm X)}
\end{align}
and it is essential to evaluate the second term on the r.h.s.\ of
\eqref{lboldsymbolayes}.  By conditioning,
\begin{align} \label{rhosigma}
  \rho_{i t}(\bm X)
  =
  \sum_{\bm\sigma:\sigma_t = i}
  \pr_{i t}(\bm\eta=\bm\sigma)
  \rho(\bm X \gv \bm\eta = \bm\sigma).
\end{align}

\begin{lem} \label{l:binary}
  Let $\bm\eta$ be a binary process on $T$.  Given $g\in C^2(\RR^T,
  \RR)$, denote its gradient and Hessian by $g'$ and $g''$,
  respectively.  Then as $\rx\to 0$,
  \begin{align}
    \ln \mean[\ex^{g(\rx\bm\eta)}]
    &=
    g(\bm 0) + [g'(\bm 0)\tp \mean(\bm\eta)]\rx
    +  \frac{1}{2}[\mean(\bm\eta\tp g''(\bm 0)
    \bm\eta) + g'(\bm 0)\tp \cov(\bm\eta) g'(\bm 0)]\rx^2 +
    o(\rx^2).
    \label{energy}
  \end{align}
  Moreover, if $\bm\pi = (\pi_t)_{t\in T}$ and $\bm
  J = (J_{s t})_{s,t\in T}$, with $\pi_t = \pr(\eta_t = 1)$ and $J_{s
    t} = \pr(\eta_s = \eta_t = 1)$, then
  \begin{align}
    \mean(\bm\eta) = \bm\pi,
    \quad
    \mean(\bm\eta\tp g''(\bm 0) \bm\eta) = \tr(g''(\bm 0) \bm J),
    \quad
    \cov(\bm\eta) = \bm J - \bm\pi\bm\pi\tp.
    \label{binary}
  \end{align}
\end{lem}

\begin{proof}
  Eq.~\eqref{energy} is well known and can be checked by standard
  calculation.  Since $\eta_t$ is binary, $\mean(\eta_t) = \pr(\eta_t
  = 1) = \pi_t$, yielding the first equation in \eqref{binary}.
  Likewise, $\mean(\bm\eta \bm\eta\tp) = \bm J$.  Then the second
  equation in \eqref{binary} follows from
  \begin{align*}
    \mean(\bm\eta\tp g''(\bm 0) \bm\eta)
    =
    \mean(\tr(\bm\eta\tp g''(\bm 0) \bm\eta))
    =
    \mean(\tr(g''(\bm 0) \bm\eta \bm\eta\tp))
    =
    \tr(\mean(g''(\bm 0) \bm\eta \bm\eta\tp))
    =
    \tr(g''(\bm 0) \bm J).
  \end{align*}
  The last equation in \eqref{binary} is also easy to verify.
\end{proof}

The desired second order Taylor series approximation to $r_t(\bm X)$
is as follows.  Once an approximate value is obtained, an approximate
value of $\pr(\eta_t =1\gv \bm X)$ can be obtained by \eqref
{e:logistic}.
\begin{thm} \label{t:post-ratio}
  Let $q$ be as in \eqref{signal-density}.  Suppose that for each $\bm
  x$, $q(\bm x, \bm\theta)$ as a function of $\bm\theta\in\RR^T$
  belongs to $C^2$.  Let $\gamma(\bm x)$ and $H(\bm x)$ be the
  gradient and Hessian of $q(\bm x, \bm\theta)$ at $\bm\theta=\bm 0$,
  respectively.  Then
  \begin{align} \label{PL-R}
    r_t(\bm X)
    &=
    \ln \frac{\pr(\eta_t = 1)}{\pr(\eta_t = 0)}
    + \gamma(\bm X)\tp [\mean_{1t}(\bm\eta) - \mean_{0t}(\bm\eta)]\rx
    \nonumber \\
    &\quad
    + \frac{1}{2} [\mean_{1t}(\bm\eta\tp H(\bm X) \bm\eta) -
    \mean_{0t}(\bm\eta\tp H(\bm X) \bm\eta)] \rx^2
    \nonumber \\
    &\quad
    +
    \frac{1}{2} \gamma(\bm X)\tp [\cov_{1t}(\bm\eta) -
    \cov_{0t}(\bm\eta)] \gamma(\bm X) \rx^2 + o_{t,\bm X}(\rx^2)
  \end{align}
  as $\rx\to 0$, where $o_{t,\bm X}(\cdot)$ means that the implicit
  constant depends on $t$ and $\bm X$.
\end{thm}
\begin{proof}
  From \eqref{rhosigma} and the assumption, $\rho_{i t}(\bm X) =
  \mean_{i t}[\rho(\bm X\gv\bm\eta)] = \mean_{i t}[\ex^{q(\bm
    X,\rx\bm\eta)}]$.  Note that $\bm X$ is fixed.  Applying Lemma
  \ref{l:binary} to $g(\cdot) := q(\bm X, \cdot)$, for $i=0,1$,
  \begin{align*}
    \ln \rho_{i t}(\bm X)
    &=
    q(\bm X,\bm 0) + [\gamma(\bm X)\tp\mean_{i t}(\bm\eta)]\rx
    \\
    &\quad
    +
    \frac{1}{2}[\mean_{i t}(\bm\eta\tp H(\bm X) \bm\eta) +
    \gamma(\bm X)\tp \cov_{i t}(\bm\eta) \gamma(\bm X)] \rx^2 +
    o_{t,\bm X}(\rx^2).
  \end{align*}
  This together with \eqref{lboldsymbolayes} yields the claimed
  expansion.
\end{proof}

\subsection{Localized dependence of data on hypotheses}

We now consider the case where the dependence of $\bm X$ on $\bm\eta$
is localized.  Specifically, suppose $\bm X = (X_t)_{t\in T}$ and the
function $q$ in \eqref{signal-density} can be written as
\begin{align} \label{q-localized}
  q(\bm x, \bm\theta) = \sum_t q_t(x_t, \theta_t).
\end{align}
Under \eqref{q-localized}, $X_t$ are independent conditionally on
$\bm\eta$.  Then from Theorem \ref{t:post-ratio}, we have the
following.
\begin{cor} \label{c:post-ratio}
  Suppose that for each $x$ and $t\in T$, $q_t(x, \theta)$ as a
  function of $\theta\in\RR$ belongs to twice continuously
  differentiable.  Let $\gamma_t(x)$ and
  $k_t(x)$ be the first and second derivatives of $q_t(x,\theta)$ at
  $\theta=0$, respectively.  Let $\gamma(\bm x) =
  (\gamma_t(x_t))_{t\in T}$ and $k(\bm x) = (k_t(x_t))_{t\in T}$.
  Then as $\rx\to 0$,
  \begin{align} \label{PL-R-localized}
    r_t(\bm X)
    &=
    \ln \frac{\pr(\eta_t = 1)}{\pr(\eta_t = 0)}
    + \gamma(\bm X)\tp [\mean_{1t}(\bm\eta) - \mean_{0t}(\bm\eta)]\rx
    \nonumber \\
    &\quad
    + \frac{1}{2} k(\bm X)\tp [\mean_{1t}(\bm\eta) -
    \mean_{0t}(\bm\eta)] \rx^2
    \nonumber\\
    &\quad
    +
    \frac{1}{2} \gamma(\bm X)\tp [\cov_{1t}(\bm\eta) -
    \cov_{0t}(\bm\eta)] \gamma(\bm X) \rx^2 + o_{t,\bm X}(\rx^2).
  \end{align}
\end{cor}
\begin{proof}
  Comparing \eqref{PL-R} and \eqref{PL-R-localized}, it suffices to
  show
  \begin{align*}
    \mean_{1t}(\bm\eta\tp H(\bm X) \bm\eta) -
    \mean_{0t}(\bm\eta\tp H(\bm X) \bm\eta)
    =
    k(\bm X)\tp [\mean_{1t}(\bm\eta) -
    \mean_{0t}(\bm\eta)].
  \end{align*}
  From \eqref{q-localized}, it is seen that $H(\bm x) = \diag(k(\bm
  x))$.  Then for $i=0,1$,
  \begin{align*}
    \mean_{it}(\bm\eta\tp H(\bm X) \bm\eta)
    =
    \mean_{it}(\tr(\bm\eta\tp H(\bm X) \bm\eta))
    &=
    \mean_{it}(\tr(\diag(k(\bm X)) \bm\eta\bm\eta\tp)
    \\
    &=
    \sum_s k_s(x_s) \pr_{it}(\eta_s=1)
    = k(\bm x)\tp \mean_{it}(\bm\eta),
  \end{align*}
  and hence the desired identity.
\end{proof}

Note that $\mean_{it}(\bm\eta)$ is the vector of $\pr(\eta_s=1\gv
\eta_t=i)$, $s\in T$.  Therefore, it is determined by the first and
second moments of $\bm\eta$.  Likewise, $\cov_{it}(\bm\eta)$ is
determined by the first three moments of $\bm\eta$.  On the other
hand, $\gamma_t(x)$ and $k_t(x)$ are determined by how $X_t$ depends
on $\eta_t$, regardless of the statistical dependence at other $s\in
T$.  Under many settings of signal processing, these quantities can be
estimated directly.  One can regard $\bm\eta$ as a message sent
through a noisy environment, and $\bm X$ the noise-corrupted message
being received.  Before the communication formally starts, one can
first estimate the moments of the distribution of $\bm\eta$.  For
example, if $\bm\eta$ is a long contiguous segment of a stationary and
ergodic binary process, then the moments may be estimated using a
sample of $\bm\eta$ by the law of large numbers.  If in
addition, there is no significant long range dependence in the
distribution of $\bm\eta$, then the evaluation of the moments, say
$\mean(\eta_0 \eta_t)$, can be restricted to $t$ within an
appropriately chosen radius from 0, which significantly reduces the
number of moments to be estimated.  On the other hand, one can use
pre-selected test signals to estimate the dependence of $X_t$ on
$\eta_t$.  If for all $t$, the dependence is the same, i.e., the
function $q_t$ is the same for all $t$, then an estimation may be
obtained from a single observation of the pair $(\bm\eta, \bm X)$,
again by the law of large numbers.

Sometimes it is more convenient to express $X_t$ as a result of
interaction between $\eta_t$ and $Z_t$, where $Z_t$ are \iid and
regarded as noise.  Specifically, suppose
\begin{align} \label{pointwise}
  \begin{array}{ll}
    X_t = \trf_t(\rx\eta_t, Z_t),
    &\text{with}
    \ \bm Z = (Z_t)_{t\in T} \text{ independent of } \bm
    \eta
    \\[1ex]
    &\text{and}\
    Z_t \text{ \iid with density } \ex^{h(z)}.
  \end{array}
\end{align}
Then the distribution of $X_t$ conditional on $\eta_t$ can be derived
from the distribution of $Z_t$ and the functional form of $\trf_t$.  We
next consider two important examples, where the signal noise
interaction is additive and multiplicative, respectively.  The
approximation for the additive interaction will be used in the next
section.

\begin{exmp}[Additive noise] \label{ex:add} \rm
  Let $\trf_t(u,z) = u+z$ for all $t$.   We then find that
  \begin{align*}
    \rho(\bm X\gv\bm \eta) =
    \prod_t \exp(h(X_t - \rx\eta_t))
    = \ex^{q(\bm X, \rx\bm\eta)},
  \end{align*}
  with $q(\bm x, \bm u) = \sum_t h(x_t - u_t)$.  Then $\gamma(\bm
  x)=-h'(\bm x)$ and $H(\bm x) =  \diag(h''(\bm x))$, so by Theorem
  \ref{t:post-ratio},
  \begin{align}
    r_t(\bm X)
    &=
    \ln\frac{\pr(\eta_t=1)}{\pr(\eta_t=0)}
    -h'(\bm X) [\mean_{1t}(\bm\eta) - \mean_{0t}(\bm\eta)]\rx
    \nonumber \\
    &\quad
    + \frac{1}{2}
    h'(\bm X)\tp[\cov_{1t}(\bm\eta) - \cov_{0t}(\bm\eta)]h'(\bm X)
    \rx^2
    \nonumber \\
    &\quad
    +
    \frac{1}{2}
    h''(\bm X)\tp[\mean_{1t}(\bm\eta) - \mean_{0t}(\bm\eta)]\rx^2
    + o_{t,\bm X}(\rx^2), \quad\text{as}\ \rx\to 0.
    \label{additiveResult}
  \end{align}
\end{exmp}

\begin{exmp}[Multiplicative noise] \label{ex:mult} \rm
  Let $\trf_t(u,z) = z \ex^{-u}$ for all $t$.  Then,
  \begin{align*}
    \rho(\bm X\gv \bm\eta)
    =
    \prod_t{\exp(\rx\eta_t)\exp(h(X_t\exp(\rx\eta_t)))}
    = \ex^{q(\bm X, \rx\bm\eta)},
  \end{align*}
  with $q(\bm x, \bm u) =  \sum_t [u_t + h(x_t \ex^{u_t})]$.  Define
  $g(x)=1+x h'(x)$ and $u(x)=x h'(x) + x^2 h''(x)$.  It is easy to get
  $\gamma(\bm x) = g(\bm x)$ and $H(\bm x) = \diag(u(\bm x))$.  Then,
  as in the additive case,
  \begin{align*}
    r_t(\bm X)
    &=
    \ln\frac{\pr(\eta_t=1)}{\pr(\eta_t=0)}
    +g(\bm X) [\mean_{1t}(\bm\eta) - \mean_{0t}(\bm\eta)]\rx
    \\
    &\quad
    + \frac{1}{2}
    g(\bm X)\tp[\cov_{1t}(\bm\eta) - \cov_{0t}(\bm\eta)] g(\bm X)
    \rx^2
    \\
    &\quad
    +
    \frac{1}{2}
    u(\bm X)\tp[\mean_{1t}(\bm\eta) - \mean_{0t}(\bm\eta)]\rx^2
    + o_{t,\bm X}(\rx^2), \quad\text{as}\ \rx\to 0.
  \end{align*}
\end{exmp}

\section{Numerical experiments} \label{s:numerical}

In this section, we report simulation studies to assess multiple
testing based on approximated conditional likelihoods in the presence
of dependence.  We compare the Bayes BH procedure and the BH
procedure.  Recall that the former is based on the approximated
conditional likelihoods of hypotheses given the data while the latter
is based on marginal $p$-values of hypotheses.

\subsection{Sampling of chain of null hypotheses} \label{ss:hmc}

In the simulations, $\bm\eta = (\enum \eta m)$ is a long binary hidden
Markov chain.  This setting is very flexible since hidden Markov
models are dense among essentially all finite-state stationary
processes \cite{kungemath_1995}.  The chain is sampled as follows.

\begin{enumerate}[topsep=1ex, itemsep=0ex, parsep=.2ex,
  leftmargin=6ex]
\item Specify a finite state space $E$, a nonempty strict subset
  $F$ of $E$, and a function $\tau: E \to \{0,1\}$ that maps $F$
  to $0$ and $E\setminus F$ to $1$.

\item Randomly sample a probability vector $\bm\pi = (\pi_s)_{s\in E}$
  and a transition matrix $\bm P$ on the states of $E$ such that
  \begin{align} \label{e:stationary}
    \bm\pi\tp \bm P = \bm\pi\tp
  \end{align}

\item Simulate a stationary Markov chain $\bm M=(\enum M m)$ on $E$
  with stationary distribution $\bm\pi$ and transition matrix $\bm P$.
  Set $\eta_t = \tau(M_t)$ for each $t$.  We will abbreviate the
  transform as
  \begin{align*}
    \bm\eta = \tau(\bm M)
  \end{align*}
  and refer to $\bm M$ as the ``parent'' Markov chain.
\end{enumerate}

In all the simulations of the study, $E=\{1, 2, 3, 4, 5\}$, $F =
\{1,2\}$, and $m = 10^5$.  The function $\tau$ is used only for the
simulation of $\bm\eta$ and not available for multiple testing.
Thus, each realization of $\bm\eta$ is regarded as a large set of
hypotheses with an unknown dependence structure.  To avoid artifacts
that may confound the comparison of multiple testing procedures, both
$\bm\pi$ and $\bm P$ are randomly sampled.  Denoting by $\psig$ the
proportion of false nulls, we need $\bm\pi$ to be such that the
proportion of 1's in $\bm\eta$ is $\psig$.  This is achieved by
setting $\pi_s = (1 - \psig) \zeta_s / \sum_{s \in F} \zeta_s$ for
$s\in F$ and $\pi_s = \psig\zeta_s/\sum_{s\not\in F} \zeta_s$ for
$s\not\in F$, with $\zeta_s$ \iid$\sim$ uniformly distributed on
(0,1).

Given $\bm\pi$, the key component of the simulation is the transition
matrix $\bm P$, which is to be sampled from positive matrices, i.e.,
matrices with all entries being positive.  Consequently, $\bm M$ is
irreducible and aperiodic.  In addition to the constraint \eqref
{e:stationary}, since we need to control the strength of dependence
within $\bm\eta$, it is necessary that the mixing rate of $\bm M$ is
controllable. As is well known, 1 is an eigenvalue of $\bm P$ with
multiplicity one, and all the other eigenvalues of $\bm P$ have
absolute values, i.e. eigenmoduli, strictly less than 1.  Sort the
eigenmoduli in decreasing order as $1=\eim_1(\bm P) > \slem(\bm P)\ge
\ldots\ge \eim_d(\bm P)\ge 0$.  It is also well-known that the second
largest eigenmodulus is an important parameter in determining the
mixing rate of a Markov Chain \cite {rosenthal_1995}. However the
convergence of a Markov chain to its long term behavior is not always
determined by a single eigenmodulus.  The precise relationship of a
Markov Chain's mixing time to the spectrum of its transition matrix is
surprisingly complicated  \cite{dia_1996, levperwil_2008}. Thus in
addition to second, here we also consider the third largest
eigenmodulus.  Since the transition matrices that we simulate in the
numerical experiments are all 5-by-5, these two parameters of a
simulated matrix provide significant information about its
spectrum.  We next describe a heuristic procedure to sample $\bm P$
that has a prescribed stationary distribution $\bm\pi$.

\subsection{Sampling of random transition matrices}\label{ss:stocmat}

The ensemble of $d\times d$ uniformly distributed transition matrices,
known as the Dirichlet Markov ensemble, can be represented by the
uniform distribution on  $\Delta_d^d$, the $d$-fold Cartesian product
of $\Delta_d = \{(\enum x d)\in [0,1]^d: \cum x d=1\}$.  The support
of this ensemble is a convex $d(d-1)$ dimensional polytope.  Uniform
sampling of $\bm P$ from the ensemble is easy because the rows of $\bm
P$ are \iid $\sim (\enum \xi d)/(\cum \xi d)$ with $\xi_i$ \iid\
exponentially distributed; see \cite{chafai_2010} for more detail.
However, in our study, we need to sample from the $(d-1)^2$
dimensional polytope of transition matrices that have a prescribed
stationary distribution $\bm \pi$.  To the best of our knowledge,
exact uniform sampling from this polytope is unknown for general
$\bm\pi$ \cite{capsombru_2009}.  An alternative to exact sampling is
to use asymptotically exact methods.  For example, \cite
{dialebmic_2012} constructs a Markov chain whose stationary
distribution is the uniform distribution on a convex polytope.  In
\cite{chadiasly_2010}, a Gibbs sampler is provided for the uniform
distribution on the set of doubly stochastic matrices, i.e.,
transition matrices whose transposes are also transition matrices.
However, doubly stochastic matrices are restrictive for our study, as
their stationary distribution is always the uniform one on 1, \ldots,
$d$.  To sample from the polytope of transition matrices with a
specified stationary distribution, we use a heuristic described
below.  It should be pointed out that matrices sampled this way are
not uniformly distributed on the polytope associated with $\bm\pi$.
Also, the complexity of the sampling increases fast as $d$ increases,
as the dimension of the polytope grows at the order of $d^2$.  We
mention that the properties of a ``typical'' matrix from the uniform
distribution on a closely related polytope of matrices was studied in
\cite{barvinok_2010}.

Sinkhorn's algorithm \cite{sinkhorn_1964, sinkhorn_1967} is an
iterative scaling procedure that maps each positive matrix, not
necessarily square, to a unique positive matrix with prescribed row
and column sums.  Its scheme is simply to scale the rows of a matrix
to obtain the prescribed row sums, then scale the columns of the
resulting matrix to obtain the prescribed column sums, and so on until
convergence.  The algorithm has been extended in many works (cf.\
\cite {bruparsch_1966, hartfiel_1971, eavcurhof_1985, marolk_1968,
  knight_2008, rotsch_1989}), and its geometric rate of convergence
and validity for irreducible nonnegative matrices have been
established \cite {fralor_1989}.  We will use the following version of
the algorithm \cite {sinkhorn_1967}, which, given a $d$-dimensional
probability vector $\bm\pi$ with all entries being positive, samples a
positive matrix $\bm A$ whose row sum and column sum vectors are both
equal to $\bm \pi$ up to a prescribed precision level.  Consequently,
the final output $\bm P$ is a transition matrix with stationary
distribution equal to $\bm\pi$ \cite{hartfiel_1974}.

\begin{enumerate}[itemsep=0ex, parsep=0ex, leftmargin=3ex, topsep=1ex]
\item[] \textbf{Sinkhorn's Algorithm}
\item[] Input: $\bm\pi\in$ Interior$(\Delta_d)$
\item[] Sample $\bm A$ uniformly from the set of $d\times d$
  transition matrices.  ($\bm A$ is positive w.p.~1.)

\item[] While $\|\bm A  \bm 1_d- \bm \pi \| > \rx$ or
  $\|\bm 1_d\tp \bm A - \bm \pi\tp \| >  \rx $

  \begin{enumerate}[itemsep=0ex, parsep=0ex, topsep=0ex,
    leftmargin=3ex]
  \item[] $\bm D^L \leftarrow [\diag(\bm A\bm 1_d)]^{-1}\diag(\bm
    \pi)$

  \item[] $\bm A \leftarrow \bm D^L \bm A$

  \item[] $\bm D^R \leftarrow [\diag(\bm 1_d \tp \bm A)]^{-1}
    \diag(\bm\pi)$

  \item[] $\bm A \leftarrow \bm A \bm D^R$

  \end{enumerate}

\item[] Return $\bm P = \diag(\bm\pi)^{-1}\bm A$
\end{enumerate}

Sinkhorn's algorithm partitions the set of irreducible nonnegative
matrices into equivalent classes, each consisting of matrices that
converge to the same transition matrix with stationary distribution
$\bm\pi$.  However, since there is no closed form characterization of
the equivalence, the sampling distribution of the algorithm in general
is unknown.  For the special case of $\bm\pi = \bm1_d/d$, the sampling
distribution is in fact not uniform but ``locally flat'' at the center
of the polytope of doubly stochastic matrices \cite {capsombru_2009}.

Despite its intractable sampling distribution, Sinkhorn's algorithm
has several advantages.  First, it allows prescribing the stationary
distribution of $\bm P$ exactly.  Second, it is computationally
tractable due to its scalability to high dimensions, geometric
convergence, and simplicity.  Third, it allows generation of
irreducible transition matrices containing one or more 0's by starting
with a nonnegative matrix with 0's in exactly the same entries \cite
{bruparsch_1966}. Consequently, it is applicable to simulation of
Markov chains on strongly connected finite directed graphs.  It is
known that under the uniform distribution on $\Delta^d_d$, all the
eigenvalues of a $d\times d$ transition matrix, except 1, converge in
probability to 0 as $d\toi$ \cite {golneu_2003, borcapcha_2012}.  Our
numerical experiments indicate that this is also the case under the
sampling distribution of Sinkhorn's algorithm.

To specify the size of the leading eigenmoduli of $\bm P$, choose
$\lambda \in [0,1)$ as a lower bound on $\slem(\bm P)$ and $\mu \in
[0,1)$, $\mu <\lambda$ as a lower bound on $\tlem(\bm P)$. For
$\lambda = 0$, we simply set $\bm P = \bm1\bm\pi\tp$, while for
$\lambda \ge 0$ and $\mu \ge 0$ we repeatedly sample $\bm P$ by
Sinkhorn's algorithm until $\slem(\bm P) \ge \lambda$ and $\tlem(\bm
P) \ge \mu$.

\subsection{Comparing Bayes BH and BH procedures}
\label{ss:compare}

After $\bm\pi$ and $\bm P$ are sampled and fixed, the simulation
proceeds in two steps.  First, one realization of $\bm\eta$ is sampled
and used to estimate $\mean(\eta_0)$, $\mean(\eta_i\gv \eta_0)$, and
$\mean(\eta_i\eta_j\gv\eta_0)$.  We denote this realization by
$\bm\theta$.  Then, an independent realization of $\bm\eta$, still
denoted by $\bm\eta$, is sampled and used to generate data $\bm X  =
\rx\bm\eta + \bm Z$, where the entries of $\bm Z$ are \iid$\sim
N(0,1)$ and $\rx>0$ is the signal strength.  In the simulation,
$\bm\eta$ can be viewed as a signal that needs to be sent through a
noisy environment.  To do this, we first estimate the moments and
conditional moments of the signal up to order two using a sample
$\bm\theta$.  Meanwhile, we estimate the distribution of the noise
$\bm Z$, for example, by sending a long sequence of 0's.  In
principle, we can have very good knowledge about the noise environment
by repeatedly testing it.  Then, based on the estimates, multiple
testing is used to recover new signals sent through the environment.

Given $\bm X$, $\pr(\eta_t=1\gv\bm X)$ are approximately evaluated
using the result in Example \ref{ex:add}.  To reduce computation, only
$\mean(\eta_s\gv \eta_t)$ and $\mean(\eta_s\eta_u\gv \eta_t)$ with
$|s-t|\le w$ and $|u-t|\le w$ are evaluated, where the ``half window
length'' $w\ge 0$ determines the number of moments that need to be
incorporated in the approximation.  All the other moments are treated
as 0.  In the absence of long range dependence, $w$ can be small.  The
approximated $\pr(\eta_t=1\gv\bm X)$ are then used by the Bayes BH
procedure.  On the other hand, the marginal one-sided $p$-values based
on individual $X_t$ are used by the BH procedure.

Let $\alpha$ be the nominal FDR control level for both the Bayes BH
and the BH procedures.  The Bayes BH procedure can estimate the
population proportion of false nulls using the observed sample $\bm
\theta$.  However, the BH procedure does not rely on $\bm \theta$.  To
level the ground of comparison, the nominal FDR control level for the
BH procedure is increased to
\begin{align*}
  \tilde \alpha = \alpha/(1-\tilde\psig)\quad
  \text{with }\tilde\psig = (\cum\eta m)/m,
\end{align*}
i.e., $\tilde\psig$ is the actual proportion of false nulls underlying
the data $\bm X = \rx\bm\eta + \bm Z$.  It is well known that the BH
procedure under the augmented FDR control level is more powerful while
still controlling the FDR at the nominal level \cite{benhoc_1995}.  In
contrast, the Bayes BH procedure has no access to statistics of $\bm
\eta$ or $\bm Z$.  It uses $\hat\psig=m^{-1} \sum \theta_t$ as the
estimate of the proportion of false nulls, and derives all the
estimates of moments and conditional moments from the observed sample
$\bm \theta$.

For a single instance of multiple testing, its performance can be
measured by the FDP and the number of true discoveries ($NTD$).  Both
measures are random variables as they are functions of the data as
well as the procedure being used.  Their population means and
variances can be used to assess the overall performance of a
multiple testing procedure, in particular, its validity, power, and
stability.  To estimate these parameters, the Bayes BH and the BH
procedures are applied to multiple replications of data, all of which
are generated with the same $\bm\eta$ but with independent
realizations of noise.  To be specific, given $\bm\eta$, the following
steps are taken.

\begin{enumerate}[itemsep=0ex, parsep=0ex, leftmargin=5ex,
  topsep=0ex]

\item Sample $\bm X = \rx \bm \eta + \bm Z$.  At nominal FDR control
  level $\alpha$, apply the Bayes BH procedure to the approximated
  $\pr(\eta_t=1\gv\bm X)$.  On the other hand, at the augmented
  nominal FDR control level $\tilde\alpha$, apply the BH procedure to
  the marginal $p$-values of $X_t$ under the null $\eta_t=0$.

\item Repeat the above step $n=1000$ times.  Compare the $\FDP$ and
  $\NTD$ of the Bayes BH procedure and those of the BH procedure,
  respectively, in terms of sample mean, sample standard error (SE),
  and empirical density.
\end{enumerate}

\subsection{Results}
First, the main parameters used in the simulations are as follows; see
Sections \ref{ss:hmc}--\ref{ss:compare} for detail.  The proportion of
false nulls $\psig = \pr(\eta_t=1)$ takes values in $\{0.05, 0.1 \}$,
the signal strength $\rx$ in $\{0.75, 1, 1.25\}$, and the nominal FDR
control level $\alpha$ in $\{0.1, 0.2, 0.3, 0.4, 0.5\}$. For the
purpose of these numerical experiments, an eigenmodulus of about 0.5
to 0.6 is considered to be moderate, an eigenmodulus of larger than
0.9 is considered to be large and an eigenmodulus smaller than 0.2 is
considered to be small.  Given half window size $w$, only
$\mean(\eta_s \gv\eta_t)$, $\mean(\eta_s\eta_u \gv \eta_t)$ with
$|s-t|\le w$ and $|u-t|\le w$ are used to approximate
$\pr(\eta_t=1\gv\bm X)$, where $\bm\eta = \tau(\bm M)$ and $\bm X =
\rx\bm\eta + \bm Z$.  If $\slem(\bm P)>0$, then $w=3$.  On the other
hand, if $\slem(\bm P) = 0$, i.e., $\eta_s$ are \iid$\sim \bm\pi$,
then $w=0$.  The length of $\bm\eta$ is $m = 10^5$.  Sample mean,
sample SE, and empirical density of the outcomes of multiple testing
are based on $n=1000$ replications of $\bm X$ with $\bm\eta$ being
fixed.

\subsubsection{Strong Dependence} \label{ss:strong}
\begin{table}
  \caption{Comparison of the Bayes BH (BBH) and the BH procedures by
    sample means and variations of FDP and NTD in the presence of
    strong dependence.  Here $\slem(\bm P)$ is large and $\tlem(\bm
    P)$ is small.  The transition matrix of the parent Markov chain is
    \eqref{e:strong} with $\slem(\bm P)=0.9074$, $\tlem(\bm P)=0.0645$
    and $\psig = 10\%$.
    \label{t:strong}}
  \begin{center}
    \begin{tabular}{|l|l||l|l||l|l|} \hline
      \multicolumn{2}{|c||}{} &
      \multicolumn{1}{c|}{BBH} & \multicolumn{1}{c||}{BH} &
      \multicolumn{1}{c|}{BBH} & \multicolumn{1}{c|}{BH} \\\hline
      $\rx$ & $\alpha$ & \multicolumn{2}{c||}{$\hat\mu_\FDP$} &
      \multicolumn{2}{c|}{$\SE(\hat\mu_\FDP)$}
      \\\hline
      0.75 & 0.4& 0.41246  & 0.40276   &   0.069156  & 0.12349 \\
      0.75& 0.5&  0.50544 &  0.49928 & 0.030706  &  0.051165 \\
      1 & 0.1& 0.0469 & 0.0999 & 0.1682 & 0.1633 \\
      1 & 0.15 &   0.1478 &  0.1465 &  0.1321 & 0.1324 \\
      1 & 0.2 & 0.2095 & 0.2043 & 0.0672 & 0.0995 \\
      1 & 0.2  & 0.2034  &  0.2010  & 0.0719  &  0.0968 \\
      1& 0.3 & 0.3085 &  0.3027 & 0.0325 & 0.0422 \\
      1.5 & 0.1& 0.10743  &  0.0994 & 0.016124 & 0.017342 \\
      1.5 & 0.2 & 0.20871 & 0.19912 & 0.010182 &
      0.012625\\\hline\hline
      $\rx$ & $\alpha$ &
      \multicolumn{2}{c||}{$\hat\mu_\NTD$} &
      \multicolumn{2}{c|}{$\SE(\hat\mu_\NTD)$}\\\hline
      0.75 & 0.4 & 29.259  &  34.002 & 6.6888  &  23.973   \\
      0.75 &0.5 & 134.2   &   131.19   &   13.801    &   57.623 \\
      1 & 0.1 &  1.0530  &  4.9680 &  1.4192 & 5.1406 \\
      1& 0.15& 8.0090 & 12.5270 & 3.8180 & 10.5292  \\
      1& 0.2 & 28.7990  &  30.8470 & 6.6985  & 18.7630 \\
      1 & 0.2 & 24.7470 & 30.1690 &  6.1466 & 18.8499 \\
      1 & 0.3 & 137.0970 & 132.3660 & 13.4961 & 43.9223 \\
      1.5 & 0.1 & 358.46 & 322.19  & 22.156 & 46.92 \\
      1.5 & 0.2 & 1096.5 &  1028.1  &  37.538 &   73.39 \\\hline
    \end{tabular}
  \end{center}
\end{table}

Table \ref{t:strong} displays the results of a set of simulations in
which $\slem(\bm P)$ and $\tlem(\bm P)$ are large and small respectively.
The parent Markov chain $\bm M$ has $\slem(\bm P) = 0.9074$ and
$\tlem=0.0645$,
\begin{align}\label{e:strong}
  \bm P =
  \begin{pmatrix}
    0.9286  &  0.0122 &   0.0276 &   0.0063 &   0.0253\\
    0.0150  & 0.8548  &  0.0651  &  0.0310  &   0.0340\\
    0.0250  &  0.9629 &   0.0097 &   0.0008 &   0.0016\\
    0.1653  &  0.7046 &   0.0955 &   0.0267 &   0.0079\\
    0.2716  &  0.6988 &   0.0257 &   0.0007 &   0.0033
  \end{pmatrix},
\end{align}
$\bm\pi = (0.3006, 0.5994, 0.0505, 0.0211, 0.0283)\tp$, and $\psig =
0.1$.  Both the Bayes BH and the BH procedures control the FDR around
the nominal level $\alpha$, with the former having somewhat smaller
variance of the FDP.  On the other hand,  while both procedures have
similar average NTD, the Bayes BH procedure has substantially less
variation in the NTD.  Thus, the approximate conditional probabilities
provide better stability than the marginal $p$-values.  In spite of
the large $\slem(\bm P)$, the half window length $w=3$ seems to work
well.  Table \ref{t:verystrong} displays the results of another set of
simulations in which $\slem(\bm P)$ is even larger, but still shows
similar phenomenon observed in Table~\ref{t:strong}.  In the
simulations, $\slem(\bm P) = 0.9566$,
\begin{align} \label{e:verystrong}
  \bm P =
  \begin{pmatrix}
    0.8993  &  0.0017522 &  0.0098301 &   0.031722 & 0.057401 \\
    0.015641 & 0.96471 & 0.0086498 & 0.0077613 & 0.0032366 \\
    0.66064 & 0.11572 &  0.075542 &  0.067761 &    0.080337 \\
    0.71019 &    0.090088 & 0.06356 &  0.13104 &   0.0051248\\
    0.57536  &   0.025292  &    0.27334  & 0.10032  & 0.025692
  \end{pmatrix}
\end{align}
with $\bm\pi = (0.67187, 0.22813, 0.02414, 0.033354, 0.042506)\tp$, and
$\psig = 0.1 $.

\begin{table}
  \caption{
    Comparison of the Bayes BH and BH procedures in the presence of
    strong dependence.  $\slem(\bm P)$ and $\tlem(\bm P)$ are again
    large and small respectively.  The transition matrix of the parent
    Markov chain is \eqref{e:verystrong} with $\slem=0.9566$,
    $\tlem=0.1155$ and $\psig = 10\%$.}
  \label{t:verystrong}
  \begin{center}
    \begin{tabular}{|l|l||l|l||l|l|} \hline
      \multicolumn{2}{|c||}{} &
      \multicolumn{1}{c|}{BBH} & \multicolumn{1}{c||}{BH} &
      \multicolumn{1}{c|}{BBH} & \multicolumn{1}{c|}{BH} \\\hline
      $\rx$ & $\alpha$ & \multicolumn{2}{c||}{$\hat\mu_\FDP$} &
      \multicolumn{2}{c|}{$\SE(\hat\mu_\FDP)$}
      \\\hline
      0.75 & 0.3 & 0.29201 &   0.3071  &  0.20283   &   0.20708
      \\
      0.75 & 0.4 & 0.39533  & 0.39839 & 0.066775  & 0.13137
      \\
      1 & 0.2 &    0.21889  &   0.20237 & 0.057124 &  0.096743
      \\
      1 & 0.2 &    0.21572  &   0.19602 & 0.056001 &   0.094159
      \\
      1&0.25& 0.25369& 0.2529 & 0.040186 &0.062352
      \\\hline\hline
      $\rx$ & $\alpha$ &
      \multicolumn{2}{c||}{$\hat\mu_\NTD$} &
      \multicolumn{2}{c|}{$\SE(\hat\mu_\NTD)$}\\\hline
      0.75 &  0.3 & 5.365  &  10.233 & 3.1205 & 9.6933
      \\
      0.75 & 0.4 & 36.139 &  35.77  &   7.9397   &  24.235
      \\
      1 & 0.2 & 44.426 & 27.889& 8.9892 & 17.99
      \\
      1 & 0.2 & 42.673   &   26.994  & 8.8113  &  17.099
      \\
      1&0.25 & 90.668 & 66.256 & 12.665  & 31.082
      \\\hline
    \end{tabular}
  \end{center}
\end{table}

In addition to summary statistics such as sample mean and sample SE,
density curves also show better stability of the Bayes BH procedure.
Figure~\ref{fig:strong} displays the kernel smoothing density curves
of the FDP and the NTD from several additional sets of simulations
with strong dependence (transition matrices not shown).  Consistent
with the above results, the distribution of the NTD of the Bayes BH
procedure is substantially more concentrated than that of the BH
procedure.

\begin{figure}
  \caption{Densities of the $\FDP$ (left) and $\NTD$ (right) of the
    Bayes BH and BH procedures in the presence of strong dependence of
    hypotheses}
  \label{fig:strong}
  \begin{center}
    \showfig{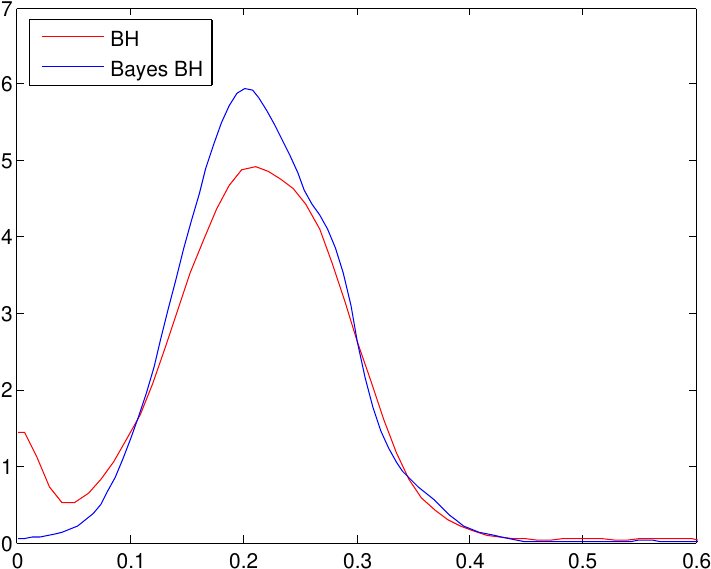}\qquad
    \showfig{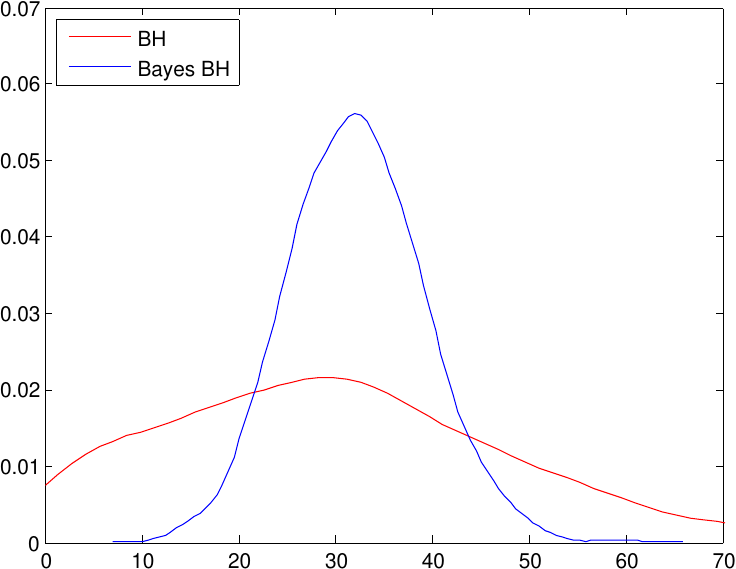}
    \\
    $\psig=0.1$, $\alpha=0.2$, $\rx=1$, $\slem=0.91$
    \\[2ex]
    \showfig{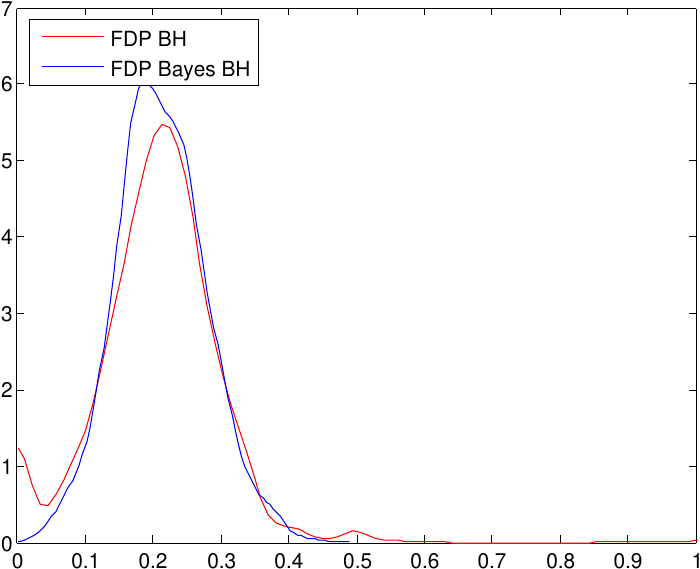}\qquad
    \showfig{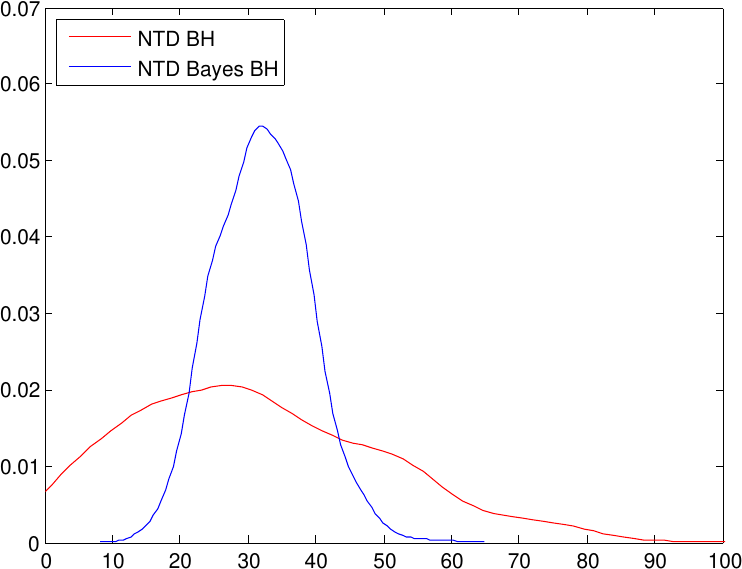}
    \\
    $\psig=0.1$, $\alpha=0.2$, $\rx=1$, $\slem=0.96$
    \\[2ex]
    \showfig{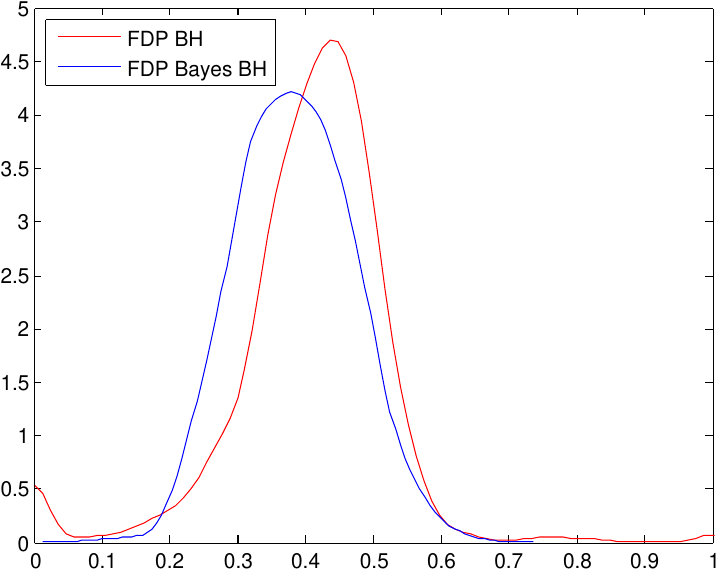}\qquad
    \showfig{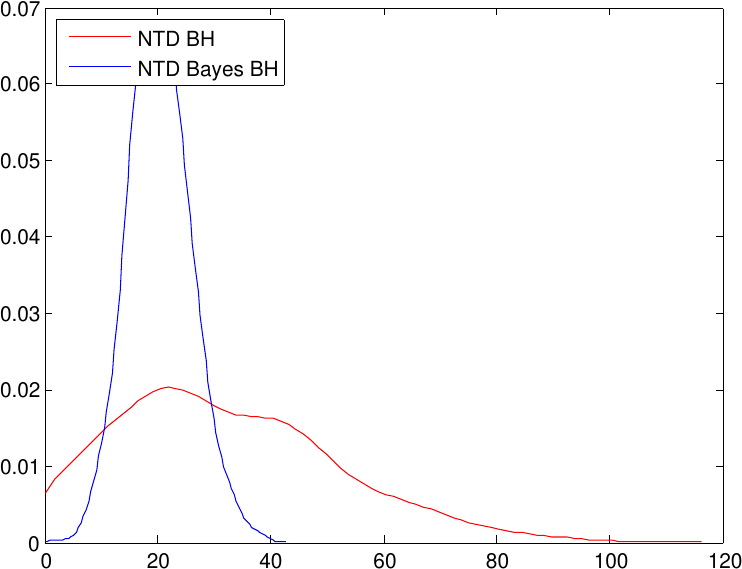}
    \\
    $\psi = 0.05$, $\alpha=0.4$, $\rx=1$, $\slem=0.95$
  \end{center}
\end{figure}
\subsubsection{Additional effects of dependence structure on multiple
  testing}
Table \ref{t:verystrong_3} displays the results of a set of
simulations in which the dependence of the hypotheses is strong with
$\slem(\bm P) = 0.9591$.  However, the third largest eigenmodulus,
i.e. $\tlem(\bm P)$, is also sampled to be relatively large at
$0.5232$.  This makes the sampling of $\bm P$ more difficult.  In
contrast, in the simulations of Section \ref {ss:strong}, there is no
control on $\tlem(\bm P)$ and as a result, $\tlem(\bm P)$ is small.
In most cases, the value is less than 0.3.  For example, for the
transition matrices in \eqref{e:strong} and \eqref{e:verystrong}, the
value of $\tlem(\bm P)$ is 0.0645 and 0.1155, respectively.  The
transition matrix of the parent Markov chain for Table
\ref{t:verystrong_3} is
\begin{align} \label{e:verystrong_3}
  \bm P =
  \begin{pmatrix}
    0.0014  & 0.9653  & 0.0124 &   0.0080 &    0.0129 \\
    0.9506  &  0.0011  &  0.0056  &  0.0155 &   0.0272\\
    0.3236   & 0.1203  &  0.5312  &  0.0238 &   0.0011\\
    0.8080  &  0.1589  &  0.0190  &  0.0025  &  0.0115\\
    0.1898 &   0.7959  &  0.0053   & 0.0024  &  0.0066\\
  \end{pmatrix}
\end{align}
with $\bm\pi = (0.47319,0.47681,0.018907,0.01173,0.019363)\tp$ and
$\psig = 0.05$.  Note that the proportion of false nulls is only half
as large as the one in Section \ref{ss:strong}, making error control
more difficult.  Table~\ref {t:verystrong_3} shows that in this case,
even though the Bayes BH procedure has a harder time to control the
FDR, it has substantially less variation in the FDP than the BH
procedure does.  Further, the power of the Bayes BH procedure is
clearly superior to the BH procedure.  For example, at a realized FDP
$\approx 0.35$, the former makes about 204 true discoveries whereas
the latter makes about 149.  The Bayes BH procedure also has
significantly lower SE-to-mean ratio of the $\NTD$ than the BH
procedure.  The phenomenon was repeatedly observed when $\tlem(\bm P)$
was sampled to be large.  Figure~\ref{fig:strong_3} provides examples
of the phenomenon.  In contrast to Figure~\ref{fig:strong},
the FDP density curves of the two procedures exhibit pronounced
differences, and their NTD density curves are far apart instead of
overlapping with each other.

\begin{table}
  \caption{
    Comparison of the Bayes BH and BH procedures in the presence of
    strong dependence with large $\slem$  and a relatively large
    $\tlem$.  The transition matrix of the parent Markov chain is
    \eqref{e:verystrong_3} with $\slem = 0.9591$, $\tlem>0.5$, and
    $\psig=5\%$.}
  \label{t:verystrong_3}
  \begin{center}
    \begin{tabular}{|l|l||l|l||l|l|} \hline
      \multicolumn{2}{|c||}{} &
      \multicolumn{1}{c|}{BBH} & \multicolumn{1}{c||}{BH} &
      \multicolumn{1}{c|}{BBH} & \multicolumn{1}{c|}{BH} \\\hline
      $\rx$ & $\alpha$ & \multicolumn{2}{c||}{$\hat\mu_\FDP$} &
      \multicolumn{2}{c|}{$\SE(\hat\mu_\FDP)$}
      \\\hline
      1 & 0.2 & 0.29271  & 0.19  & 0.10223   &  0.23451 \\
      1 & 0.3 &  0.3676 & 0.30591 &  0.06527  &  0.19306 \\
      1 & 0.4 &   0.47105  & 0.40827 &  0.039127 &  0.11613 \\
      1.25 & 0.2 & 0.35611 &  0.1958  & 0.032559 &   0.086935 \\
      1.25 & 0.356 & 0.42145 & 0.35488 & 0.02103 & 0.041823
      \\\hline\hline
      $\rx$ & $\alpha$ &
      \multicolumn{2}{c||}{$\hat\mu_\NTD$} &
      \multicolumn{2}{c|}{$\SE(\hat\mu_\NTD)$}\\\hline
      1 &  0.2 & 23.739   & 9.0709 & 3.655 &    4.1191 \\
      1 & 0.3 &  53.006  &     13.288  &  11.094  &  9.1999  \\
      1 & 0.4 & 119.15    &   19.097  &     30.862  &  18.862 \\
      1.25 & 0.2 & 204.12 &  30.266  &   23.961  &  16.364 \\
      1.25 & 0.356 & 402.27  &  149.84 &  31.26  & 37.498 \\\hline
    \end{tabular}
  \end{center}
\end{table}

\begin{figure}
  \caption{Densities of the $\FDP$ and $\NTD$ of the Bayes BH and BH
    procedures when both $\slem$ and $\tlem$ are large}
  \label{fig:strong_3}
  \begin{center}
    \showfig{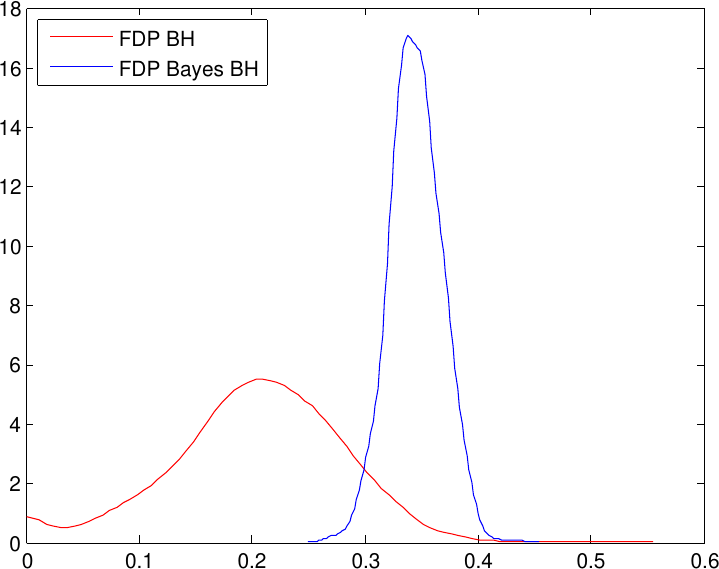}\qquad
    \showfig{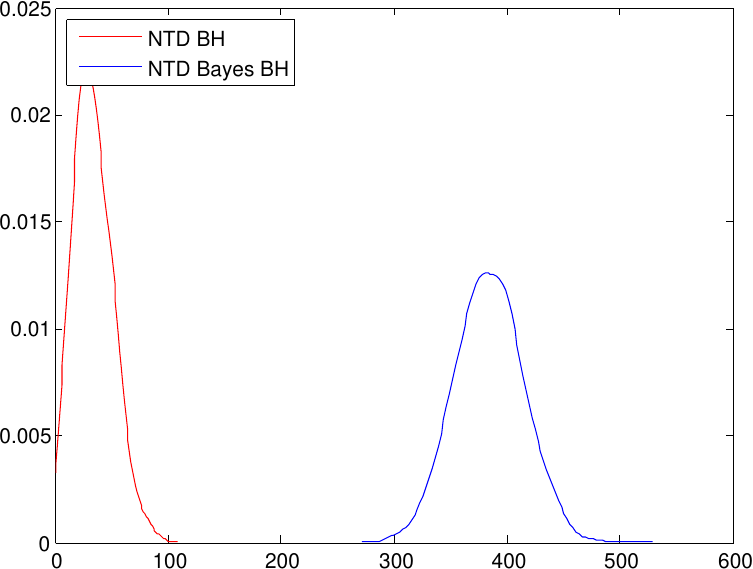} \\
    $\psi = 0.05$, $\alpha=0.2$, $\rx=1.25$, $\slem=0.91$,
    $\tlem=0.382$ \\[2ex]
    \showfig{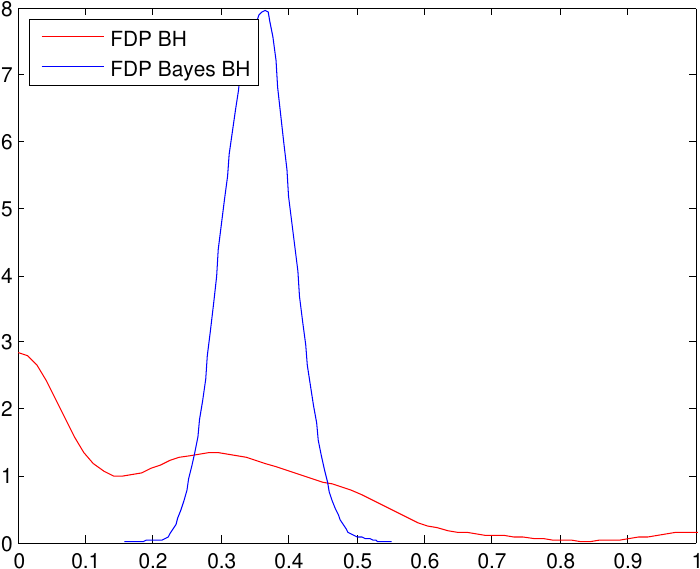}
    \qquad
    \showfig{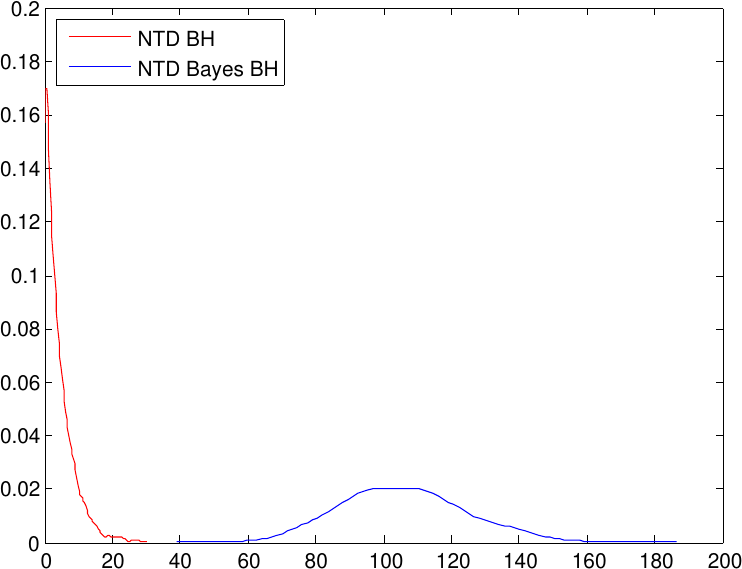} \\
    $\psi = 0.05$, $\alpha=0.2$, $\rx=1$, $\slem=0.904$,
    $\tlem=0.463$ \\[2ex]
    \showfig{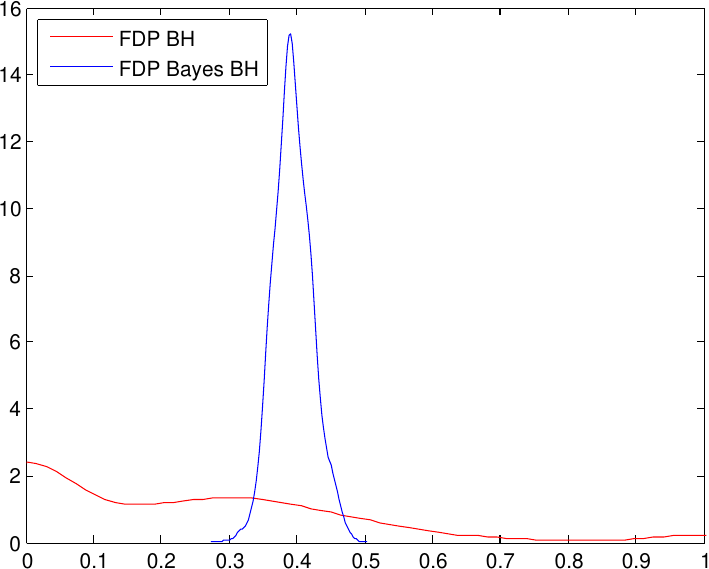}
    \qquad
    \showfig{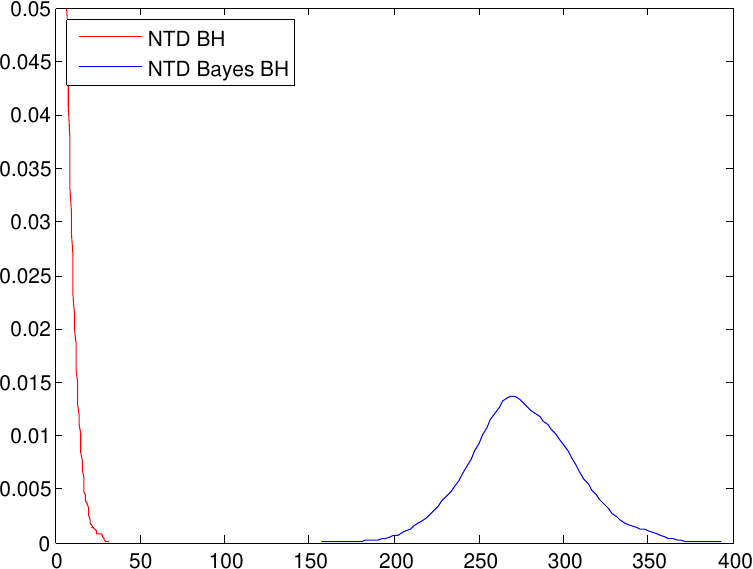}
    \\
    $\psi = 0.05$, $\alpha=0.2$, $\rx=1$, $\slem = 0.904$,
    $\tlem=0.59$
  \end{center}
\end{figure}

As noted in the introduction, it is known that the BH procedure is
often able to control the FDR under dependence, and certain properties
of dependence are sufficient for this to happen \cite{wu_2008,
  benyek_2001, sarkar_2002, sarkar_2006, clahal_2009}.  Much less is
known about the relationship between properties of dependence and
other aspects of performance of multiple testing, such as stability
and power. One advantage of random sampling from a large space of
dependence structures is that it enables experimental investigation of
the relationship.  The properties of statistical dependence as
characterized by $\slem(\bm P)$ and $\tlem(\bm P)$ have been the focus
here.  The above results indicate that $\slem(\bm P)$ alone is not
sufficient to characterize the relative performances of the Bayes BH
and the BH procedures. This numerical observation is consistent with
what is now known about the mixing times of Markov chains \cite
{dia_1996, levperwil_2008}.  The results in Table \ref{t:verystrong_3}
and Figure~\ref {fig:strong_3} indicate that $\tlem(\bm P)$ may
sometimes play a role in shaping the differences.  In our simulations,
$\bm P$ is a $5\times 5$ matrix. Thus $\slem(\bm P)$ and $\tlem(\bm
P)$ together provide a significant part of information on the spectrum
of the transition matrix.  This points to a complex relationship
between the dependence structure of the underlying  hypotheses and the
relative performances of the Bayes BH and the BH procedures.

\subsubsection{Moderate dependence}
Table \ref{t:moderate} displays results for the case of moderate
dependence.  In this set of simulations, $\slem(\bm P)=0.6019$,
\begin{align} \label{e:moderate}
  \bm P=
  \begin{pmatrix}
    0.1999  &  0.6548 &   0.0537 &   0.0561 &   0.0354 \\
    0.8932  &  0.0765 &   0.0137 &   0.0088 &   0.0076 \\
    0.4431  &  0.4365 &   0.0721 &   0.0284 &   0.0199 \\
    0.1159  &  0.8271 &   0.0371 &   0.0129 &   0.0070 \\
    0.6080 &   0.0597 &   0.0505 &   0.0710 &   0.2108 \\
  \end{pmatrix}
\end{align}
with $\bm\pi = (0.4969, 0.4031, 0.0376, 0.0349, 0.0275)\tp$ and
$\psig= 0.1$.  From the table, the $\NTD$ of the Bayes BH procedure is
clearly more stable than that of the BH procedure.  This can also be
seen from the densities of the $\NTD$ shown in Figure~\ref
{fig:moderate}.  To a lesser extent, the $\FDP$ of the Bayes BH
procedure is also more stable than that of the BH procedure except
when $\rx = 0.75$, $\alpha = 0.3$ or $0.4$.  In the latter situation,
the variance of the $\FDP$ of the Bayes BH procedure is still
competitive with that of the BH procedure.

\begin{table}
  \caption{Comparison of the Bayes BH and BH procedures in the
    presence of moderate dependence.  The transition matrix of the
    parent Markov chain is \eqref{e:moderate} with $\slem=0.6019$ and
    $\psig = 10\%$}
  \label{t:moderate}
  \begin{center}
    \begin{tabular}{|l|l||l|l||l|l|} \hline
      \multicolumn{2}{|c||}{} &
      \multicolumn{1}{c|}{BBH} & \multicolumn{1}{c||}{BH} &
      \multicolumn{1}{c|}{BBH} & \multicolumn{1}{c|}{BH} \\\hline
      $\rx$ & $\alpha$ & \multicolumn{2}{c||}{$\hat\mu_\FDP$} &
      \multicolumn{2}{c|}{$\SE(\hat\mu_\FDP)$}
      \\\hline
       0.75 & 0.3 & 0.29201 &   0.3071  &  0.20283   &   0.20708
       \\
       0.75 & 0.4 & 0.22492 & 0.3014  & 0.2719  &   0.2124
       \\
       0.75 & 0.5 &  0.4966 &  0.50118 &  0.03134  &  0.045887
       \\
       1 & 0.2 &   0.192   &    0.2012 &   0.075112  &  0.084956
       \\
        1 & 0.3 & 0.294  &  0.30074  & 0.035074   &  0.040486
      \\\hline\hline
      $\rx$ & $\alpha$ &
      \multicolumn{2}{c||}{$\hat\mu_\NTD$} &
      \multicolumn{2}{c|}{$\SE(\hat\mu_\NTD)$}\\\hline
       0.75 & 0.3 &2.094  &  9.431 &  1.8608  &  9.2073
      \\
      0.75 & 0.4 &   22.28    &  38.658  &  5.6593  &  25.412
      \\
       0.75 & 0.5 & 121.26 & 138.26  &  12.793 & 60.366
       \\
      1 & 0.2 &  22.542 &  32.796   &  5.8157   &  18.614
      \\
      1 & 0.3 &  116.64  &  129.34  &  12.984   &  42.037
      \\\hline
    \end{tabular}
  \end{center}
\end{table}

\begin{figure}
  \caption{Densities of the FDP and NTD of the Bayes BH and BH
    procedures in the presence of moderate dependence of hypotheses}
  \label{fig:moderate}
  \begin{center}
    \showfig{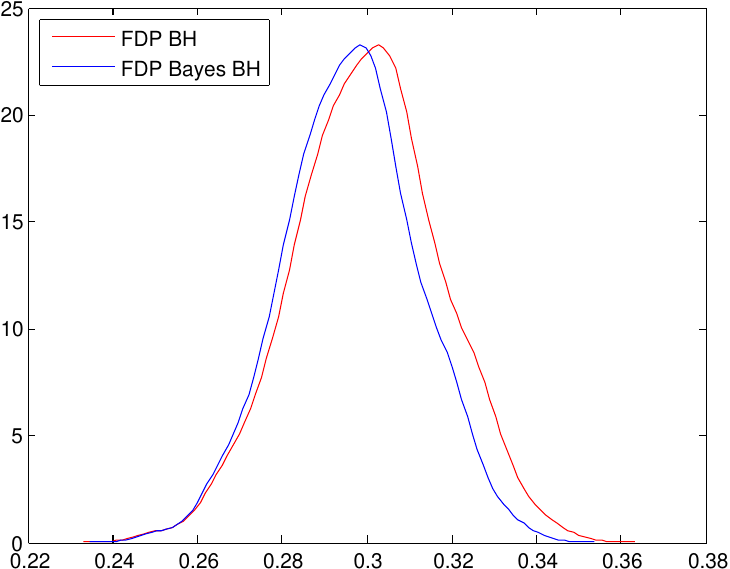}\qquad
    \showfig{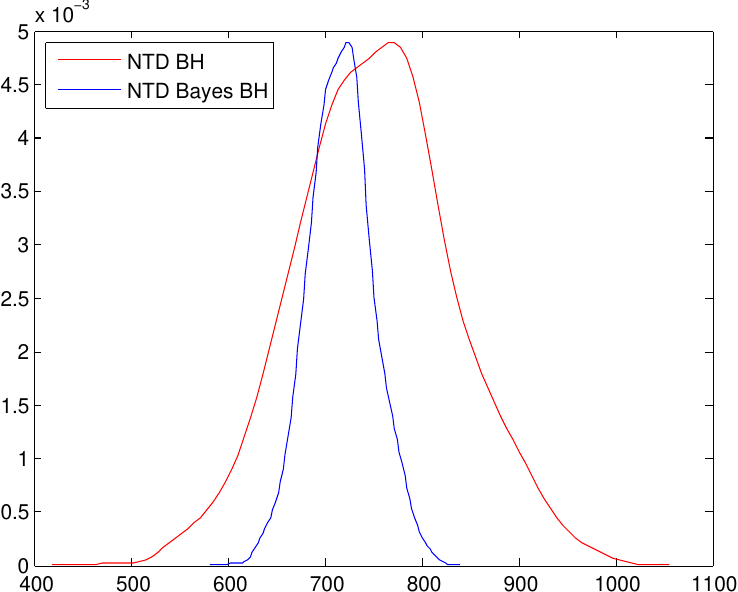}
    \\
    $\psig=0.1$, $\alpha=0.3$, $\rx=1.25$, $\slem=0.67$
  \end{center}
\end{figure}

\subsubsection{Independence} \label{ss:ind}

Finally, as a test on the validity of the Bayes BH procedure,
Table \ref{t:indep} displays results when $\slem(\bm P)=0$ and, as a
result, the entries of $\bm\eta$ are \iid.  In the simulations,
$\psig = 0.05$ and $\bm P = \bm1_5\bm\pi\tp$, where the stationary
probability vector $\bm\pi$ is randomly sampled as in Section
\ref{ss:hmc}.  As can be seen, under independence, both the Bayes BH
procedure and  the BH procedure are valid.  In terms of power,
although the NTD of the Bayes BH procedure has a smaller average than
that of the BH procedure, its stability is clearly superior.  Similar
phenomenon can be seen from the density plot in
Figure~\ref{fig:moderate}.

\begin{table}
  \caption{Comparison of the Bayes BH and BH procedures when there is
    no dependence in data}
  \label{t:indep}
  \begin{center}
    \begin{tabular}{|l|l||l|l||l|l|} \hline
      \multicolumn{2}{|c||}{} &
      \multicolumn{1}{c|}{BBH} & \multicolumn{1}{c||}{BH} &
      \multicolumn{1}{c|}{BBH} & \multicolumn{1}{c|}{BH} \\\hline
      $\rx$ & $\alpha$ & \multicolumn{2}{c||}{$\hat\mu_\FDP$} &
      \multicolumn{2}{c|}{$\SE(\hat\mu_\FDP)$}
      \\\hline
      1 & 0.3 & 0.29028 &  0.29932 &     0.19437  &   0.18847\\
      1 & 0.4 &  0.39102  & 0.39308  &     0.081393  &   0.123 \\
      1 & 0.5 & 0.50282  &     0.50157  &  0.039488  &   0.065857\\
      1.25 & 0.2& 0.19442  &    0.19951 & 0.072535  &  0.083292
      \\\hline\hline
      $\rx$ & $\alpha$ &
      \multicolumn{2}{c||}{$\hat\mu_\NTD$} &
      \multicolumn{2}{c|}{$\SE(\hat\mu_\NTD)$}\\\hline
      1 & 0.3 & 5.322    &   11.16 &  2.9969         &   9.6522 \\
      1 & 0.4 &  22.881    &   32.613  &  5.9599    &    19.899 \\
      1 & 0.5 & 79.951    &   87.779  & 10.701       &    36.661 \\
      1.25 & 0.2 & 24.838  &   31.121  &  5.9329     &    15.893
      \\\hline
    \end{tabular}
  \end{center}
\end{table}

\begin{figure}
  \caption{Densities of the FDP and NTD of the Bayes BH and BH
    procedures when hypotheses are independent}
  \label{fig:ind}
  \begin{center}
    \showfig{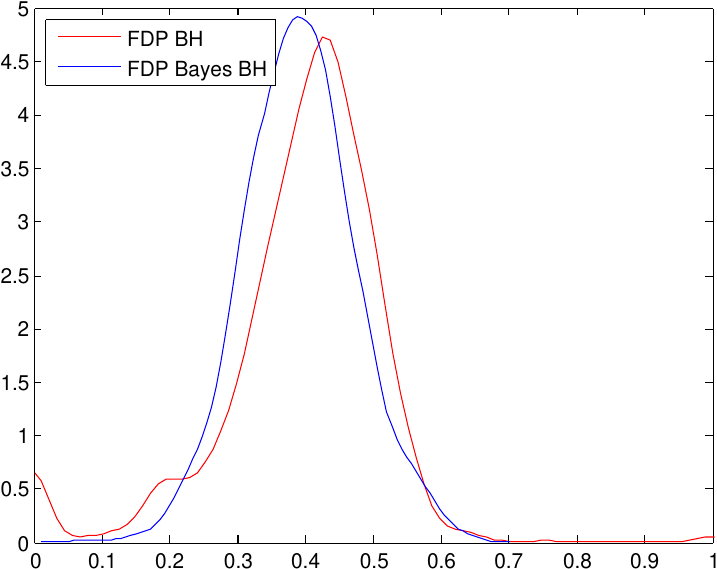}
    \qquad
    \showfig{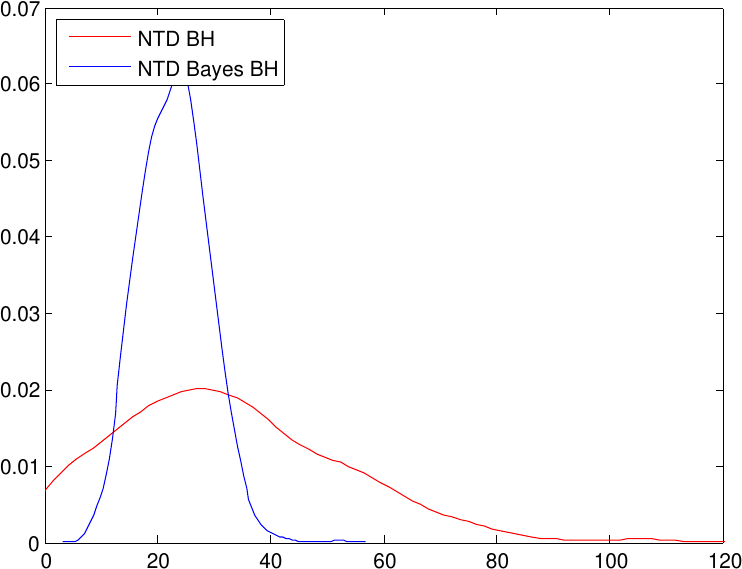}
    \\
    $\psi = 0.05$, $\alpha=0.4$, $\rx=1$, $\slem=0$ \\[2ex]
    \showfig{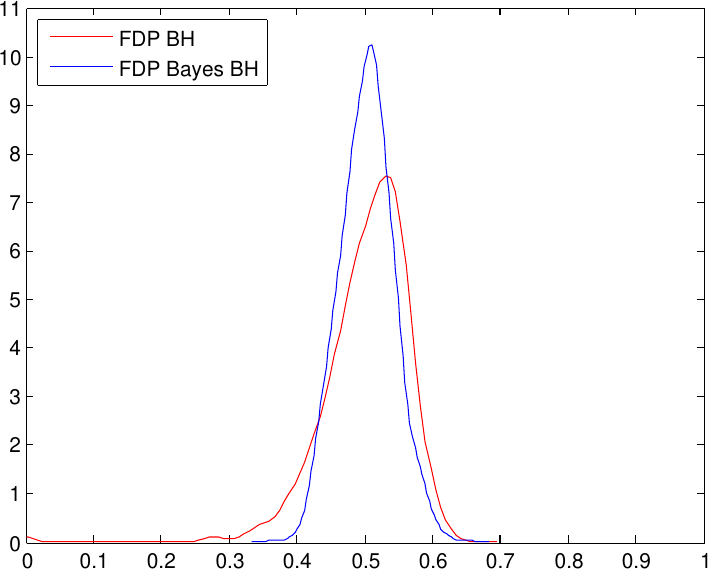}\qquad
    \showfig{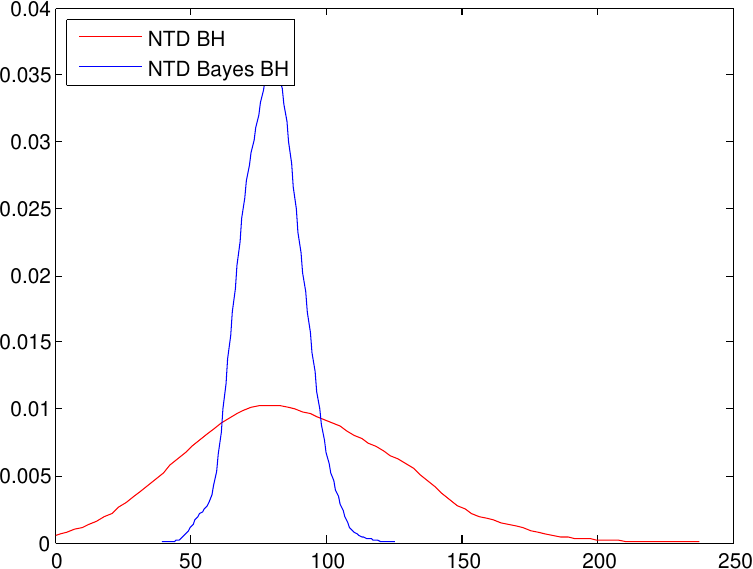}
    \\
    $\psi = 0.05$, $\alpha=0.5$, $\rx=1$,  $\slem=0$
  \end{center}
\end{figure}

\section{Conclusion} \label{s:conclusions}
We presented a flexible signal processing model of large scale
multiple testing and used probability approximation to conduct
multiple testing.  This allowed us to empirically demonstrate that
even when very little is known about the dependence structure of the
hypotheses, substantial reduction in variance can be achieved by using
approximated conditional probabilities in multiple testing.
Using simulations designed to generate very general dependence
structures, we provided evidence for the effectiveness of the local
conditioning approach to FDR and power variance reduction in large
scale  multiple testing.  Conditional likelihood has been fruitfully
exploited in multiple testing \cite{efron_2001, efron_2004,
  efron_2007}.  Indeed, the notion of Local False Discovery Rate
introduced in \cite{efron_2001} is a conditional likelihood that
does not directly borrow strength from nearby observations.
In contrast, our approach explicitly incorporates nearby
observations into the conditional likelihood.  Our analysis indicates
that the reduction in the standard deviations of the FDP and
the number of true discoveries can be substantial when conditional
likelihood is used.  In particular, for very noisy data, even when
there is some increase in the FDR as compared to using $p$-values the
reduction in the standard deviation of FDR may justify a Bayesian
approach.

The effects of dependence on large scale multiple tests has been a
topic of research at least since the beginning of this century
\cite{benyek_2001}. Various special dependence structures between the
hypotheses have been studied and continue to be studied to ascertain
the effectiveness of both $p$-value based and Empirical Bayes methods
\cite{benyek_2001, wu_2008, stephens_2017} in multiple testing.
In addition, with some notable exceptions  \cite{owen_2005, fanhangu_2012,
qiukleyak_2005, schlin_2011,gorglaqiuyak_2007, li_2015}, the standard
deviations of the FDR and the number of true discoveries has usually
taken a second place to the estimates themselves. However, given that
it is often impractical to repeat a multiple test numerous times, not
controlling the standard deviation of the power and the FDR can lead
to highly biased estimates of FDR.  The need for simulating
general dependence structures is crucial to testing the
performance of our approach. This led us to the random HMM dependence
structures that can be generated by using random Markov matrices using
Sinkhorn's algorithm.

Our approach strikes a compromise between generality and tractability
by using the signal processing framework that models data as a
function of signal, signal strength, i.i.d.\ noise and signal noise
interaction.  The flexibility can be useful in application.  For
instance, additive noise is not the only signal-noise interaction that
is seen in fMRI practice \cite{mapach_2013}, so a different
signal-noise interaction function may improve the analysis.  Assuming
a fixed signal noise interaction and a fixed noise distribution is not
unrealistic in many applications.  The challenge is that they are
usually unknown.  Two important issues that we have not addressed here
is modeling real data with this approach and obtaining a ``noiseless''
sample of the data for measuring moments.  These two issues are
intimately connected to the application in question.  Obtaining a
noiseless sample of the signals is not very difficult in a signal
processing context.  However, in DNA microarray analysis and in fMRI
it remains a challenge to not just find a noiseless sample but also to
arrive at the signal strength and the signal noise interaction. A
possible solution might be to incorporate a parametric model to the
underlying binary process that identifies false nulls.  Thus our next
goal is the development of ideas that successfully apply the signal
processing and local dependence approach to multiple testing.

\begin{hide}
The connection of large scale multiple testing to signal processing is
obvious but the challenge in applications is often that very little is
known about the signal distribution, signal dependence, the noise
distribution and how they interact.  However this is not always the
case \cite{lindquist_2008, linmej_2015}.  Stochastic models of fMRI
data have been modeled in the signal processing framework using random
fields \cite{schgavadl_2011, chesch_2017}. The fMRI blood-oxygen-level
dependent (BOLD) response is the signal of interest.  Multiple types
of noise  including system-related instabilities, subject motion and
physiological fluctuation corrupt the signal \cite{lindquist_2008}.
The ``signal-to-noise-ratio''  captures  the signal strength; see
Figure 5 in  \cite{linmej_2015}.

We have developed a novel flexible stochastic modeling framework for
casting an arbitrary multiple test as a signal processing problem.  It
explicitly decomposes the data into four distinct inputs. A Bernoulli
signal process or network, a noise with an arbitrary continuous
distribution, a scalar signal strength and a signal-noise interaction
function. Each of these features is relevant to fMRI data. We place
very few assumptions on the dependence structure of the signals while
allowing local correlation. Our probality model allows an arbitrary,
but known, continuous noise distribution.   Our model allows an
arbitrary, but known, signal-noise interaction function.
\end{hide}

\bibliography{approx-pl}

\end{document}